\documentclass{amsproc}

\usepackage{amssymb}
\usepackage{amsmath,amscd}
\usepackage{graphicx}
\usepackage[cmtip,all]{xy}
\usepackage{mathtools}
\usepackage{color}
\usepackage{ dsfont }
\usepackage{pdfsync}
\usepackage{hyperref}

\newtheorem{theorem}{Theorem}[section]
\newtheorem{lemma}[theorem]{Lemma}
\newtheorem{prop}[theorem]{Proposition}
\newtheorem{cor}[theorem]{Corollary}

\theoremstyle{definition}
\newtheorem{definition}[theorem]{Definition}
\newtheorem{example}[theorem]{Example}
\newtheorem{observation}[theorem]{Observation}

\theoremstyle{remark}
\newtheorem{remark}[theorem]{Remark}

\begin{document}

\title{Spin Structures of Kac-Moody Type}

\author{Amir Farahmand Parsa}
\address{School of Mathematics, Institute for Research in Fundamental Sciences (IPM),  P.O. Box: 19395-5746, Tehran, Iran}
\email{a.parsa@ipm.ir}
\thanks{}


\date{}

\begin{abstract}
We study spin structures on affine Kac-Moody symmetric spaces and obtain sufficient conditions for their existence.\ As a byproduct of this, we obtain a spin-c representation of certain Kac-Moody quadratic subgroups of type E.
\end{abstract}

\maketitle

\section{Introduction}
We study symmetric spaces in the context of affine Kac-Moody geometry.\ Our main goal is 
to construct certain spin structures for such symmetric spaces.\ This study can be considered as a continuation of \cite{WFrei09,WFrei14,zbMATH06479353}.\ 
An approach towards studying Kac-Moody symmetric spaces is studied in \cite{zbMATH07224322,gruning:2018} which is more topological and algebraic in nature.\ Here the approach is analytical and we consider the holomorphic completions of these spaces, called \textit{affine Ka-Moody symmetric spaces}.\ They were initially studied in \cite{WFrei09}.\ 

These symmetric spaces 
have interesting connections with supergravity theories.\ It is known among M-theorists  
that the factor space $E_{10}/K(E_{10})$ which is the base for certain bosonic geodesic $\sigma$-models
has links to bosonic dynamics in some 11-dimensional supergravity theory (see \cite{zbMATH06186809}).\ In this context, elements of the quadratic subgroup or simply ``compact'' form $K(E_{10})$ of the real algebraic hyperbolic Kac-Moody group
$E_{10}$ are called \textit{hidden Kac-Moody symmetries}.\ Among these symmetries, those in
$K(E_{9})$ can be understood better because of their close connection with finite dimensional compact forms of finite dimensional semisimple Lie groups.\
For the bosonic theories, it is known that the affine factor space $E_{9}/K(E_{9})$ can be used to describe the space of solutions for motion equations for $D=11$ supergravity possessing nine commuting Killing vectors (see \cite{zbMATH03994467, zbMATH05051744}).\ Furthermore, in connection with quantum field theory, Dirac-like operators on such affine paces can be interpreted as supercharge of the supersymmetry in nonlinear $\sigma$-model in 1+1 dimensions (see \cite{zbMATH03787630, zbMATH04048607}).\ Moreover, having a spinor bundle on such spaces can help to classify and understand solutions of Dirac-like wave equations arising in supergravity theories as it does in the finite dimensional cases.

Constructed in \cite{WFrei09} (see also \cite{WFrei14,zbMATH06479353}),\ affine Kac-Moody symmetric spaces are natural infinite dimensional generalization of
finite dimensional symmetric spaces.\ They come from the holomorphic completion of algebraic
affine Kac-Moody groups and enjoy a natural tame Fr\'echet structure (for tame Fr\'echet spaces see \cite{zbMATH03787692}).\
This tame Fr\'echet structure enables us to use functional analytical tools such as the inverse function theorem
and implicit function theorem to analyze affine Kac-Moody symmetric spaces.\ Affine Kac-Moody symmetric spaces under study here are generally of the form $G/K$ where $G$ is a (real) complex geometric affine Kac-Moody group and $K$ a geometric (``compact'' real) real form obtained as the fixed point subgroup of a Cartan-Chevalley involution on $G$.\ Since their associated subalgebras are  \textit{quadratic subalgebras} of the ambient Kac-Moody algebra, we also called those fixed point subgroups \textit{quadratic subgroups}.\ Also note that these fixed point subgroups are not compact but since their finite dimensional counter parts are called compact forms, here we borrow this name from the finite dimensional terminology. 

By \textit{geometric} affine Kac-Moody groups we mean the holomorphic
completion of minimal algebraic affine Kac-Moody groups (for minimal algebraic Kac-Moody groups see \cite{zbMATH01821146,zbMATH06916409}).\ More precisely, let $G_{\mathds{R}}$ be a linear semisimple Lie group with a complexification $G_{\mathds{C}}$.\ We denote the corresponding algebra to $G_{\mathds{R}}$ and $G_{\mathds{C}}$ with $\mathfrak{g}_{\mathds{R}}$ and $\mathfrak{g}_{\mathds{C}}$ respectively.\ In this setting, we consider $\mathds{C}^{*}:=\mathds{C}\backslash\{0\}$ as the unique linear complexification of the circle $\mathcal{S}^{1}$.\ Define
\begin{equation*}
\mathds{C}^{*}G^{\sigma}_{\mathds{C}}:=\{f:\mathds{C}^{*}\to G_{\mathds{C}}~|~ f~\text{holomorphic and}~\sigma\circ f(z)=f(\omega z)\},
\end{equation*}
and 
\[
\mathds{C}^{*}G^{\sigma}_{\mathds{R}}:=\{f\in\mathds{C}^{*}G^{\sigma}_{\mathds{C}}~|~f(\mathcal{S}^{1})\subset G_{\mathds{R}}\},
\]
where $\sigma$ is a diagram automorphism of order $n$ of $\mathfrak{g}_{\mathds{C}}$ and $\omega:=e^{\frac{2\pi i}{n}}$.\ A complex \textit{geometric affine Kac-Moody group} $\widehat{\mathds{C}^{*}G^{\sigma}_{\mathds{C}}}$ is a certain $(\mathds{C}^{*})^{2}$-bundle
over $\mathds{C}^{*}G^{\sigma}_{\mathds{C}}$ whose natural real form $\widehat{\mathds{C}^{*}G^{\sigma}_{\mathds{R}}}$ is a certain torus fiber bundle over
$\mathds{C}^{*}G^{\sigma}_{\mathds{R}}$ (for details see \cite{WFrei09}).\ When $\sigma$ is the identity then $\widehat{\mathds{C}^{*}G^{\sigma}_{\mathds{R}}}$ is referred to as \textit{untwisted} and denoted by $\widehat{\mathds{C}^{*}G_{\mathds{R}}}$.\ Note that construction is analogous to the construction of affine Kac-Moody groups as certain two dimensional
extensions of polynomial loop groups.\ Now for a compact real form $K_{\mathbb{R}}$ of $G_{\mathbb{R}}$,
we define $\widehat{\mathds{C}^{*}K^{\sigma}_{\mathds{C}}}$ and $\widehat{\mathds{C}^{*}K^{\sigma}_{\mathds{R}}}$ as above.\ It turns out that these groups
can be considered as natural ``compact'' (real) forms of $\widehat{\mathds{C}^{*}G^{\sigma}_{\mathds{C}}}$
and $\widehat{\mathds{C}^{*}G^{\sigma}_{\mathds{R}}}$ respectively (see \cite{WFrei09}).\ Let $G$
denote a geometric affine Kac-Moody group and $K$ its ``compact'' (real) form for a fixed compact
(real) form $K_{\mathbb{R}}$ of $G_{\mathbb{R}}$.\ By the main results of \cite{WFrei09} (see also \cite{WFrei14,zbMATH06479353}) we know that all of the above groups are tame Fr\'echet Lie groups
with a topology locally compatible with the compact-open topology on $\mathds{C}^{*}G^{\sigma}_{\mathds{C}}$.\ Moreover, it is also shown in \cite{WFrei09} that the corresponding
quotient spaces, $G/K$ which are called \textit{affine Kac-Moody symmetric spaces}, are tame Fr\'echet manifolds.\ 

In Section~\ref{secFacSpPB}, we give a version
of the implicit function theorem for tame Fr\'echet spaces which enables us to a generalize the methods in \cite{zbMATH03098679} and prove that the natural projection 
\[
\pi:G\to G/K
\]
is a \textit{locally trivial} tame Fr\'echet principal $K$-bundle.\ In \cite{WFrei09}, the geometry considered on $G/K$ is induced by the tame Fr\'echet geometry on $G$ via $\pi$.\ Now having local trivializations 
for the bundle $\pi$ leads us to a better understanding of the local geometry of $G/K$ by looking
at $G$.\ 

In the context of string theory, physicists consider particles as loops on the space-time.\ This leads to studying free loop spaces and consequently free loop groups (see e.g., \cite{zbMATH04020059}).\  
There has been a lot of progress in the study of loop spaces and groups (in the sense of \cite{zbMATH04216267} and \cite{Land99}).\ By the local triviality for $\pi$ we are now able to pull-back structures on loop spaces over our symmetric spaces which are again locally trivial.\ In contrast with loop spaces of compact Lie groups, there exists the concept of duality among affine Kac-Moody symmetric spaces.\ Moreover, geometric affine Kac-Moody groups as holomorphic completions of algebraic minimal affine Kac-Moody groups (for algebraic minimal Kac-Moody groups see \cite{zbMATH04017219} and \cite{zbMATH01821146}) are carefully
chosen so that they contain important symmetries such as isometries of the algebraic Kac-Moody algebra in the algebra level (see \cite[Remark 3.13]{WFreyn11}).\ The relation between symmetric loop spaces
and our affine Kac-Moody symmetric spaces is not obvious.\ 

In Section~\ref{secELB} we show that for affine Kac-Moody symmetric spaces of
compact type, there exists a continuous injection from our symmetric spaces to symmetric loop spaces.\ In this section we also provide a method by which one can produce certain \textit{loop-string structures} on symmetric spaces coming from their finite dimensional part.\ 

An interesting $\text{SO}(32)$ representation of certain fixed point subgroups of Kac-Moody type is introduced and studied in a series of articles in mathematical physics (see \cite{zbMATH06088567,DaKlNi006,BuyHePa06,zbMATH02212653}).\ 
In Section~\ref{secAKMSSCT} we define a \textit{finite spin (spin-c) structure} on affine Kac-Moody symmetric spaces
of type $\widetilde{E}_{n}$ ($n\leq8$) via the above $\text{SO}(32)$ representation of
their ``compact'' forms.\ Again by the local triviality obtained in Section~\ref{secFacSpPB}, one can define the {\v C}ech cohomology and benefit from its features specially
its cohomological spectral sequence.\ Having the {\v C}ech cohomology for these bundles, we can
provide certain conditions for existence of finite spin (spin-c) structures on
the underlying symmetric spaces.\ A spin (or spin-c) structure on the symmetric spaces under study here provides a lift of the above $\text{SO}(32)$ representation to $\text{Spin}(32)$ (or $\text{Spin}^{c}(32)$).

The author would like to thank Walter Freyn, Ralf K\"ohl and Karl-Hermann Neeb for their 
valuable comments on the preliminary draft.

\section{Tame Fr\'{e}chet principal bundles}\label{secFacSpPB}
Let $G$ be a geometric affine Kac-Moody group as described in the introduction.\ Let $K$ be the fixed-point set of a Chevalley-Cartan involution and $G/K$ the corresponding affine Kac-Moody symmetric space.\ The main purpose of this section is to show that the natural projection 
\begin{equation}\label{eqthmmaincoor}
\pi:G\to G/K,
\end{equation}
constitutes a locally trivial $K$-principal bundle in the context of tame Fr\'{e}chet geometry (see Theorem~\ref{maintheorel} below).\ For this we need some fundamental results in functional analysis adjusted to our context.\ Our main reference for such results are \cite{WFrei09} and \cite{zbMATH03787692}.\ We start this section by citing a famous inverse function theorem due to Nash and Moser. 
\begin{theorem}[Nash-Moser Inverse Function Theorem]\label{nashmoser}\cite[Thoerem III.1.1.1]{zbMATH03787692}
	Let $F$ and $G$ be two tame Fr\'{e}chet spaces.\ Let $\phi:U\subseteq F\to G$ be a smooth tame map where $U$ is an open subset of $F$.\
	Suppose that the equation for the derivative $D_{f}\phi(h)=k$ has a unique solution $h=V_{f}\phi(k)$ for all $f\in U$ and all $k$
	and the family of inverses $V\phi:U\times G\to F$ is a smooth tame map.\ Then $\phi$ is locally invertible and each local inverse $\phi^{-1}$ is a smooth tame map.
\end{theorem}
There are many different versions of the implicit function theorem derived
from the above Nash-Moser inverse function theorem (see e.g. \cite[\S\S III.3.3]{zbMATH03787692}).
Here we give a proof for a simpler version of this theorem which is essentially \cite[Theorem 2.1.3]{WFrei09} (see also \cite[Theorem 2.48]{WFreyn11}) but more
suitable for the purpose of this article.
\begin{cor}[Implicit Function Theorem]\label{nashim}\cite[Theorem 2.1.3]{WFrei09}
	Let $F$, $V$ and $W$ be three tame Fr\'{e}chet spaces.\ Let $\phi:U\subset F\times V\to W$ be a smooth tame map where $U$ is open.\
	Suppose that the partial derivative $D_{y}\phi(x,y)$ has a unique smooth tame right inverse for all $(x,y)\in U$.\ If $\phi(x_{0},y_{0})=0$ for some
	$(x_{0},y_{0})\in U$, then there exists open sets $U'$ and $U''$ in $F$ and $V$ respectively and a unique smooth tame map $\psi:U'\to U''$ such that $\phi(x,\psi(x))=0$.
\end{cor}
\begin{proof}
	Define
	\begin{equation} \label{nashimeq1}
	\begin{cases}
	\Phi: & U\subseteq F\times V\to F\times W \\
	& (x,y)\mapsto(x,\phi(x,y)).
	\end{cases}
	\end{equation}
	By \cite[Lemma 2.1.5]{WFrei09}, $\Phi$ is a smooth tame map (see also \cite[\S II.2]{zbMATH03787692}).\ Moreover, $D\Phi$ can be represented in the following matrix form with respect to the decomposition of the domain and range of $\Phi$:
	\begin{equation} \label{nashimeq2}
	\left(
	\begin{array}{cc}
	Id & D_{x}\phi \\
	0 & D_{y}\phi \\
	\end{array}
	\right).
	\end{equation}
	By hypotheses, the partial derivation $D_{y}\Phi(x,y)$ has a unique tame right linear inverse, $D^{-1}_{y}\Phi$, for all $(x,y)\in U$.\ Define $D^{-1}\Phi:U\subseteq F\times V\to F\times W$ by the following representation:
	\begin{equation} \label{nashimeq3}
	\left(
	\begin{array}{cc}
	Id & -D_{x}\phi D_{y}^{-1}\phi \\
	0 & D_{y}^{-1}\phi \\
	\end{array}
	\right),
	\end{equation}
	which defines a unique right inverse for $D\Phi$.\ Note that since $D_{y}^{-1}\phi:W\to V$, the composition $-D_{x}\phi D_{y}^{-1}\phi$ in (\ref{nashimeq3}) is valid.\ Moreover, since $\Phi$ is a smooth tame map, the family of inverses $D^{-1}\Phi:U\times F\times V\to F\times W$ is smooth and tame and linear with respect to its $F$ and also $V$ components and hence it is smooth and tame with respect to all of its components by \cite[Theorem 3.1.1.]{zbMATH03787692}.
	
	Now the Nash-Moser inverse function theorem, Theorem~\ref{nashmoser}, applied to $\Phi$ implies that $\Phi$ is locally invertible and following similar arguments as in the proof of the classical implicit function theorem, we conclude the corollary (see, e.g., the 2nd part of the proof of \cite[Theorem 2.1.3]{WFrei09}).
\end{proof}

\begin{definition}[Slicing]\label{slicing}
	Let $G$ be a topological group and $K$ be a closed subgroup.\ Endow $G/K$ with the quotient topology.\ A \textit{slicing} $\psi$ for the element $K$ in $G/K$ is a continuous function defined from a neighborhood $N$ around $K$ to $G$ such that $\psi(K)$ is mapped to the identity in $G$ and $\psi$ remains constant on every coset $xK\in N$.\ In other words, $\psi$ is a continuous cross section of the projection map $\pi:G\to G/K$.
\end{definition}
\begin{definition}[Tame Fr\'{e}chet submanifold]\label{tamesubmanifold}
	Let $M$ be a tame Fr\'{e}chet manifold.\ A closed subset $N$ of $M$ is called a \textit{tame Fr\'{e}chet submanifold}\ if for every $x\in N$ there exists a chart $(\phi_{x},U_{x})$ of $M$ containing $x$ such that $\phi_{x}:U_{x}\to F\times W$ for a pair of tame Fr\'{e}chet spaces $F$ and $W$, and in addition,
	\[
	y\in N\cap U_{x}~~ \text{if and only if}~~ \phi_{x}(y)\in 0\times W.
	\]
\end{definition}
Inspired by \cite{zbMATH03098679}, here we prove a slicing function theorem
by means of our version of the implicit function theorem in the tame Fr\'echet context.
\begin{prop}\label{slicingexists}
	Let $G$ be a tame Fr\'{e}chet-Lie group.\ Let $K$ be a closed tame Fr\'{e}chet-Lie subgroup of $G$ such that the quotient space $G/K$ forms a tame Fr\'{e}chet manifold with respect to the quotient topology.\ Then there exists a slicing for $K$. In addition, when the projection is smooth and tame, we obtain a smooth tame slicing of $K$.
\end{prop}
\begin{proof}
	Let $(\phi,U)$ be a submanifold chart of $G$ for $K$ at the identity $e\in G$.\ By definition, there are two tame Fr\'{e}chet spaces $F$ and $W$ such that
	\[
	\phi:U\to F\times W=:H.
	\]
	We can assume that $\phi(e)=0\in H$ and since $G$ is a tame Fr\'{e}chet-Lie group, we can also assume that $U\cdot U\subseteq U$.\ Moreover, from the definition submanifold charts we know that $x\in K\cap U$ if and only if $\phi(x)\in 0\times W$.\\
	Define $\widetilde{U}:=\phi(U)$, $X:=\phi(x)$ for all $x\in U$ and
	\begin{equation} \label{slicingexistseq1}
	\begin{cases}
	&\Phi:\widetilde{U}\times \widetilde{U}\subseteq H\times H \to H=F\times W \\
	& (X,Y)\mapsto(\Phi_{1}(X,Y),\Phi_{2}(X,Y)),
	\end{cases}
	\end{equation}
	where
	\begin{equation} \label{slicingexistseq2}
	\begin{cases}
	&\Phi_{1}:\widetilde{U}\times \widetilde{U}\to F\\
	&(X,Y)\mapsto\pi_{F}\circ\phi(x^{-1}y)=:\phi(x^{-1}y)_{F},
	\end{cases}
	\end{equation}
	and
	\begin{equation} \label{slicingexistseq3}
	\begin{cases}
	&\Phi_{2}:  \widetilde{U}\times \widetilde{U}\to W \\
	&(X,Y)\mapsto\pi_{W}(Y)=:Y_{W},
	\end{cases}
	\end{equation}
	where $\pi_{F}$ and $\pi_{W}$ are the natural projections of $H$ onto $F$ and $W$ respectively.\ Note that all these maps are clearly smooth and tame (\cite[\S II.2]{zbMATH03787692}).\\
	In the above setting, the partial derivative $D_{Y}\Phi(X,Y)$ has the following matrix representation:
	\begin{equation} \label{slicingexistseq4}
	\left(
	\begin{array}{cc}
	D_{Y}\Phi_{1}^{X}\\
	D_{Y}\Phi_{2} \\
	\end{array}
	\right),
	\end{equation}
	where
	$\Phi_{1}^{X}:=\Phi(X,-)$.\\
	First note that $D_{Y}\Phi_{2}:H\to W$ is just the projection onto $W$ with a unique smooth tame right inverse $Di:W\hookrightarrow F\times W$.\\
	Let $DL_{X}:H\to F\times W$ be the derivative of the left translation in $G$ by $x^{-1}\in U$ which is a smooth tame linear isomorphism.\ Note that since $G$ is a tame Fr\'{e}chet-Lie group, the left and right translations are tame diffeomorphisms as well as the inverse map.\ We ignore the role of the inverse map in the definition of $\Phi$ for the same reason and simplicity.\ Then we obtain the following matrix representation for $DL_{X}$:
	\begin{equation} \label{slicingexistseq5}
	\left(
	\begin{array}{cc}
	(DL_{X})_{F}\\
	(DL_{X})_{W}\\
	\end{array}
	\right).
	\end{equation}
	Therefore, $(DL_{X})_{F}= D_{Y}\Phi_{1}^{X}$ has a unique smooth tame right inverse, namely $D_{Y}^{-1}\Phi_{1}^{X}$.\ This enables us to define a unique smooth tame right inverse $D_{Y}^{-1}\Phi(X,Y)$ for $D_{Y}\Phi(X,Y)$ for all $X,Y\in\widetilde{U}$ with the following matrix representation:
	\begin{equation} \label{slicingexistseqinverse0}
	\left(
	\begin{array}{cc}
	D_{Y}^{-1}\Phi_{1}^{X} & Di\\
	\end{array}
	\right).
	\end{equation}
	By assumption, for $X_{0}:=\phi(e)$ we have $\Phi(X_{0},X_{0})=0$.\ Now our version of the implicit function theorem, Theorem~\ref{nashim}, implies that there exists a unique smooth tame map $\Psi:\widetilde{U}'\subseteq\widetilde{U}\to\widetilde{U}''\subseteq\widetilde{U}$ such that
	\begin{equation}\label{slicingexistseqinverse1}
	\Phi(X,\Psi(X))=0.
	\end{equation}
	We claim that the map
	\begin{equation} \label{slicingfunctioneq0}
	\begin{cases}
	&\psi:\phi^{-1}(\widetilde{U}')\to\phi^{-1}(\widetilde{U}'')\\
	&\psi:=\phi^{-1}\circ\Psi\circ\phi,\\
	\end{cases}
	\end{equation}
	defines the desirable smooth tame slicing function for $K$.\\
	To see this, first note that for any $X,Y\in\widetilde{U}$ with the property $\Phi(X,Y)=0$ we have
	\begin{itemize}
		\item $\phi(x^{-1}y)_{F}=0~~\Rightarrow~~x^{-1}y\in K\cap U~~\Rightarrow~~y=xk$ for some $k\in K$;
		\item $\phi(y)_{G}=0~~\Rightarrow~~y\in G\backslash K$.
	\end{itemize}
	We conclude from above that by the uniqueness of $\Psi$, $y$ is the only element in $xK\cap G\backslash K$. And hence, $\psi$ remains constant along $xK$ for every $x$ in the domain of $\psi$.\ Therefore, by reduction of the domain, $\psi:V\to G$ can be defined where $V:=\pi(\phi^{-1}(\widetilde{U}'))$ which remains an open subset of $G/K$ containing $K$ with respect to the quotient topology on $G/K$.
\end{proof}
\begin{definition}[Tame Fr\'{e}chet fibre bundle]\label{tamebundle}
	Let $\pi:P\to M$ be a (topological) fibre bundle between tame Fr\'{e}chet manifolds.\ We call $P$ a tame Fr\'{e}chet fibre bundle over $M$ if $\pi$ is a smooth tame map and the bundle satisfies the following condition:\\
	Every $x\in M$ contains in a chart domain $(\phi_{x},U_{x})$ with values in an open subset $U$ of a tame Fr\'{e}chet space $F$ such that there exist a tame Fr\'{e}chet manifold $N$ (the typical fibre) and a smooth tame map $\psi:\pi^{-1}(U_{x})\to U\times N$ such that the following diagram commutes:
	\begin{equation} \label{tamebundlediag}
	\xymatrix{
		\pi^{-1}(U_{x}) \ar[d]_{\pi} \ar[r]^{\psi} &U\times N\ar[d]^{\pi_1}\\
		U_{x} \ar[r]^{\phi} &U}
	\end{equation}
	and $\pi^{-1}(y)\cong N$ as tame Fr\'{e}chet manifolds for every $y\in U_{x}$.\ A tame Fr\'{e}chet fibre bundle is locally trivial if $\psi$ is a tame diffeomorphism.\ A tame Fr\'{e}chet fibre bundle is called a tame Fr\'{e}chet principal $N$-bundle when the typical fibre $N$ is a tame Fr\'{e}chet-Lie group.
\end{definition}
\begin{prop}\label{propprincbundle}
	Let $G$ be a tame Fr\'{e}chet-Lie group.\ Let $K$ be a tame Fr\'{e}chet-Lie subgroup of $G$ such that the quotient space $G/K$ forms a tame Fr\'{e}chet manifold with the quotient topology.\ Then the natural projection $\pi:G\to G/K$ forms a locally trivial principal $K$-bundle.\ Moreover, if $\pi$ is a smooth tame map then it is a locally trivial tame Fr\'{e}chet principal $K$-bundle over $G/K$.
\end{prop}
\begin{proof}
	Proposition~\ref{slicingexists} and \cite[Theorem 1\'{}]{zbMATH03098679} imply that $\pi:G\to G/K$ is a principal $K$-bundle.\\
	To show that $\pi$ is a locally trivial tame Fr\'{e}chet fibre bundle by the homogeneity of $G/K$ it suffices to show that for a chart $(\phi,U_{e})$ at $K$ of $G/K$, the tame Fr\'{e}chet condition (\ref{tamebundlediag}) holds for some $\psi:\pi^{-1}(U_{e})\to U\times K$.\ Under the identification $\pi^{-1}(x)\cong_{tame}x\times K$ define
	\begin{equation}\label{trivialization00}
	\psi(x,k)=(x,s(x)^{-1}\cdot k),
	\end{equation}
	where $s$ denotes the smooth tame cross section of $K$ obtained in Proposition~\ref{slicingexists}.\ Since by the hypothesis $\pi$ is a smooth tame map, we conclude that $\psi$ is a smooth tame map with a smooth tame inverse defined as follows:
	\begin{equation}\label{trivialization01}
	\psi^{-1}(x,k)=s(x)\cdot k.
	\end{equation}
	This completes the proof.
\end{proof}
\begin{observation}\label{comformsunmnfld}
	Let $G/K$ be an affine Kac-Moody symmetric space as defined in the introduction (see \cite{WFrei09} for details).\ Then $K$ is a tame Fr\'{e}chet submanifold of $G$ by means of the explicit charts given in \cite[Chapters 3 and 4]{WFrei09}.\ Moreover, the construction of such charts also shows that $\pi:G\to G/K$ is a smooth tame submersion in the sense of \cite[Definition 4.4.8]{zbMATH03787692}.
\end{observation}
\begin{theorem}\label{maintheorel}
	Let $G/K$ be an affine Kac-Moody symmetric space.\ Then
	\begin{equation}\label{eqthmmain}
	\pi:G\to G/K,
	\end{equation}
	is a locally trivial tame Fr\'{e}chet principal $K$-bundle.\ In particular, if $G/K$ is of non-compact type, then $G\cong G/K\times K$.
\end{theorem}

\begin{proof}
	The theorem is a direct consequence of Proposition~\ref{propprincbundle} and Observation~\ref{comformsunmnfld}.\ In particular, when $G/K$ is of non-compact type by \cite[Corollary 4.4.1]{WFrei09}, it is diffeomorphic to a topological vector space and hence contractible.\ Because $\pi:G\to G/K$ is a locally trivial principle bundle by the first part of the theorem, when $G/K$ is of non-compact type then $G\cong G/K\times K$ as topological spaces.
\end{proof}


\begin{cor}\label{maincor}
	Let $G/K$ be an affine Kac-Moody symmetric space.\ Then
	the topological gauge group of $\pi$, namely, $\text{Aut}_{G/K}(G)$ is isomorphic to the topological group
	of cross sections of the  locally trivial tame Fr\'{e}chet principal $K$-bundle
	\begin{equation}\label{eqthmmaincor01}
	\text{Ad}(\pi):\text{Ad}(G):=G\times_{\text{Ad}}K\to G/K,
	\end{equation}
	where $\text{Ad}(\pi)$ is constructed by considering the conjugation action of $K$ on itself.\\
	Moreover, if $G/K$ is of non-compact type, then $\text{Aut}_{G/K}(G)$ is isomorphic to the topological mapping group $\text{Map}(G/K,K)$.
\end{cor}
\begin{proof}
	Since, by Theorem~\ref{maintheorel}, $\pi$ is a locally trivial tame Fr\'{e}chet principal $K$-bundle
	and $K$ is a tame Fr\'{e}chet-Lie group it follows that $\text{Ad}(\pi)$ is indeed a locally trivial tame Fr\'{e}chet principal $K$-bundle.\ Let $\Gamma(\text{Ad}(\pi))$ denote the topological group
	of cross sections of $\text{Ad}(\pi)$.\ Then
	\[
	\text{Aut}_{G/K}(G)\cong\Gamma(\text{Ad}(\pi))\]
	is an immediate consequence of \cite[Remark 7.1.4 and Remark 7.1.6]{zbMATH00496152}.\ Moreover, when $G/K$ is of non-compact type, then by Theorem~\ref{maintheorel}, $\pi$ is trivial and hence the isomorphism
	\[
	\text{Aut}_{G/K}(G)\cong\text{Map}(G/K,K)
	\]
	follows from \cite[Proposition 7.1.7]{zbMATH00496152}.
\end{proof}

\section{Loop spin structures}\label{secELB}
Affine Kac-Moody algebras and groups have celebrated loop realizations (see e.g., \cite[Chapter 7]{zbMATH00194085}).\ Inspired by this, we establish a connection between geometric affine Kac-Moody groups and smooth loop groups.\ By doing so, we manage to introduce and construct certain loop structures on affine Kac-Moody symmetric spaces.\ The main reference we use for the concepts mentioned here related to loop groups and algebras is \cite{zbMATH04002404}.\ We also frequently use the notion of compact-open $C^{\infty}$-topology.\ The main references that we use here for this notion and related concepts are \cite{zbMATH05503386,zbMATH05294708}.
 
Let $G_{\mathds{R}}$ be a linear connected semisimple compact Lie group.\ It essentially follows form the Peter-Weyl theorem that $G_{\mathds{R}}$
posses a unique (up to isomorphism) complexification $G_{\mathds{C}}$ (see \cite[\S 2.3]{zbMATH04002404}).\ Therefore, in this setting, we always consider $\mathds{C}^{*}:=\mathds{C}\backslash\{0\}$ as the unique linear complexification of the circle $\mathcal{S}^{1}$.\ Define
\[
\mathds{C}^{*}G_{\mathds{C}}:=\{f:\mathds{C}^{*}\to G_{\mathds{C}}~|~ f~\text{holomorphic}\},
\]
and 
\[
\mathds{C}^{*}G_{\mathds{R}}:=\{f:\mathds{C}^{*}\to G_{\mathds{C}}~|~ f~\text{holomorphic~and~}f(\mathcal{S}^{1})\subset G_{\mathds{R}}\}.
\]
Let $LG_{\mathds{R}}$ denote the free loop group of $G_{\mathds{R}}$.\ Recall from \cite[Chapter 3]{zbMATH04002404} that $LG_{\mathds{R}}$ is an infinite dimensional, paracompact manifold
modeled on $L\mathds{R}^{n}$ with respect to the (compact-open) $C^{\infty}$-topology, i.e., the topology of 
uniform convergence of the functions and all their partial derivatives of 
all orders.\ Moreover, with its natural group structure, and the above topology, $LG_{\mathds{R}}$ constitutes a loop group (see also \cite[Definition 1.1]{zbMATH05294708})).\ Now since $\mathds{C}^{*}$ is the complexification of $\mathcal{S}^{1}$, we define 
\begin{equation}\label{eq1to1}
\begin{cases}
H:&\mathds{C}^{*}G_{\mathds{R}} \hookrightarrow LG_{\mathds{R}} \\
&f\mapsto f|_{\mathcal{S}^{1}},
\end{cases}
\end{equation}

\begin{lemma}\label{lemcontembpre}
	$H$ as defined above is a continuous injection.
\end{lemma}
\begin{proof}
	By the uniqueness theorem for holomorphic functions, any $f\in\mathds{C}^{*}G_{\mathds{R}}$ is completely and uniquely determined by
	its values on $\mathcal{S}^{1}$.\ This implies that $H$ is well-defined and injective.
	
	To see that $H$ is continuous first note that, by \cite[Theorem 4.7]{WFreyn11}, the Fr\'echet topology on $\mathds{C}^{*}G_{\mathds{R}}$ is (locally)
	compatible with the compact-open topology.\ And the compact-open topology coincides
	with the $C^{\infty}$-topology for holomorphic functions (cf.  \cite{zbMATH05294708}, see also \cite[\S I.5]{zbMATH05503386}).\ But the topology
	on $LG_{\mathds{R}}$ is also the $C^{\infty}$-topology, hence continuity. 
\end{proof}

Let $\widehat{\mathds{C}^{*}G_{\mathds{R}}}/ \widehat{\mathds{C}^{*}K_{\mathds{R}}}$ be an affine
Kac-Moody symmetric space of compact type defined in \cite[\S 4.3]{WFrei09} (or see \cite[\S 5.4]{WFreyn11}).\ We call the tame Fr\'{e}chet manifolds $\mathds{C}^{*}G_{\mathds{R}}/ \mathds{C}^{*}K_{\mathds{R}}$,  \textit{primitive affine Kac-Moody symmetric spaces}.\ Recall that here $K_{\mathds{R}}$ is the fixed point subgroup of $G_{\mathds{R}}$ associated to a Chevalley Cartan involution on $G_{\mathds{R}}$.\ Here we always assume that $K_{\mathds{R}}$ is connected.\ Let $H$ be as in (\ref{eq1to1}), we consider the following commutative diagram of topological embeddings:
\begin{equation} \label{homeoloopp}
\xymatrix{
	\mathds{C}^{*}G_{\mathds{R}} \ar@{^{(}->}[r]^H & LG_{\mathds{R}}\\
	\mathds{C}^{*}K_{\mathds{R}} \ar@{^{(}->}[r]^H \ar@{^{(}->}[u] & LK_{\mathds{R}} \ar[u]\ar@{^{(}->}[u]}
\end{equation}
This yields the following continuous mapping for the corresponding quotient spaces:
\begin{equation}\label{eqliftablles01477}
\mathds{C}^{*}G_{\mathds{R}}/\mathds{C}^{*}K_{\mathds{R}}\to LG_{\mathds{R}}/LK_{\mathds{R}}.
\end{equation}
\begin{lemma}\label{liftlopembed}
	The continuous mapping (\ref{eqliftablles01477}) is injective.
\end{lemma}
\begin{proof}
	Let $\gamma$ and $\delta$ be two loops in $\mathds{C}^{*}G_{\mathds{R}}$
	such that $[\gamma]=[\delta]$ in $LG_{\mathds{R}}/LK_{\mathds{R}}$ via $H$ as in (\ref{eq1to1}).\ Hence there exists a loop $\alpha$ in $LK_{\mathds{R}}$ such that $\gamma\cdot\alpha=\delta$ on $\mathcal{S}^{1}$.\
	Since $\gamma$ and $\delta$ are both in $\mathds{C}^{*}G_{\mathds{R}}$, $\alpha=\delta\cdot\gamma^{-1}$ has a canonical extension $\tilde{\alpha}$ to $\mathds{C}^{*}$.\ This implies that $[\gamma]^{\mathds{C}^{*}}=[\delta]^{\mathds{C}^{*}}$ in
	$\mathds{C}^{*}G_{\mathds{R}}/\mathds{C}^{*}K_{\mathds{R}}$.\ Hence (\ref{eqliftablles01477}) defines a continuous injection.
\end{proof}
\begin{lemma}\label{freeloopfunctorgothru}
	With the above notations, we have the following continuous injection
	\begin{equation}\label{eqfreeloopfunctorphi00000}
	LG_{\mathds{R}}/LK_{\mathds{R}}\hookrightarrow L(G_{\mathds{R}}/K_{\mathds{R}}),
	\end{equation}
	with respect to the (compact-open) $C^{\infty}$-topology.
\end{lemma}
\begin{proof}
	Define
	\begin{equation}\label{eqfreeloopfunctorphi1}
	\begin{cases}
	\Phi:&LG_{\mathds{R}}\to L(G_{\mathds{R}}/K_{\mathds{R}})  \\
	&\begin{cases}
	\gamma:&\mathcal{S}^{1}\to G_{\mathds{R}}  \\
	&t\mapsto \gamma(t)
	\end{cases} \mapsto\begin{cases}
	\Phi(\gamma):&\mathcal{S}^{1}\to G_{\mathds{R}}/K_{\mathds{R}}\\
	&t\mapsto [\gamma(t)]
	\end{cases}.\\
	\end{cases}
	\end{equation}
	Since the mapping $\pi:G_{\mathds{R}}\to G_{\mathds{R}}/K_{\mathds{R}}$ is continuous ,
	$\Phi$ is continuous.\\
	Moreover, $\Phi$ is constant on $LK_{\mathds{R}}$ as it maps any loop in $LK_{\mathds{R}}$ to
	the constant loop $t\mapsto K_{\mathds{R}}$ in $L(G_{\mathds{R}}/K_{\mathds{R}})$.\ Therefore, we obtain
	the following continuous map induced by $\Phi$:
	\[
	\tilde{\Phi}:LG_{\mathds{R}}/LK_{\mathds{R}}\to L(G_{\mathds{R}}/K_{\mathds{R}}).
	\]
	Now, if $\Phi(\gamma)=\Phi(\delta)$ for a pair $\gamma,\delta\in LG_{\mathds{R}}$ then
	\[
	\forall t\in\mathcal{S}^{1}~:~~[\gamma(t)]=[\delta(t)]~\Longrightarrow~~\gamma(t)K_{\mathds{R}}=\delta(t)K_{\mathds{R}}
	\]
	hence for all $t\in\mathcal{S}^{1}$, $\delta(t)^{-1}\gamma(t)\in K_{\mathds{R}}$.\ Define
	\[
	\begin{cases}
	\alpha:&\mathcal{S}^{1}\to K_{\mathds{R}}  \\
	&t\mapsto\delta(t)^{-1}\gamma(t),
	\end{cases}
	\]
	then it follows that $[\gamma]^{L}=[\delta]^{L}$ via $\alpha$ in $LG_{\mathds{R}}/LK_{\mathds{R}}$.
\end{proof}
Now we are able to construct $LK_{\mathds{R}}$-bundle on primitive affine Kac-Moody symmetric spaces.

\begin{prop}\label{proLKbun}
	Let $\mathds{C}^{*}G_{\mathds{R}}/ \mathds{C}^{*}K_{\mathds{R}}$ be a primitive affine Kac-Moody symmetric space where $G_{\mathds{R}}$ is a linear connected semisimple compact Lie group and $K_{\mathds{R}}$ is connected.\ Then there exists a $LK_{\mathds{R}}$-bundle on $\mathds{C}^{*}G_{\mathds{R}}/ \mathds{C}^{*}K_{\mathds{R}}$.
\end{prop}
\begin{proof}
By Lemma~\ref{freeloopfunctorgothru} the following commutative diagram exists:
\begin{equation} \label{eqfreeloopfunctorphi2}
\xymatrix{
	&LG_{\mathds{R}} \ar@{>}[d]_{\Pi}\ar@{>}[r]^{\Phi}
	&L(G_{\mathds{R}}/K_{\mathds{R}}) \\
	&LG_{\mathds{R}}/LK_{\mathds{R}}\ar@{>}[ru]_{\tilde{\Phi}}}
\end{equation}
Now we consider the following bundle:
\begin{equation}\label{eqliftablebundles022}
LG_{\mathds{R}}\to LG_{\mathds{R}}/LK_{\mathds{R}},
\end{equation}
which is an $LK_{\mathds{R}}$-bundle considered as infinite dimensional manifolds with the structure groups as infinite dimensional Lie groups in the sense of \cite{zbMATH04002404}.\\
Also, by applying the loop ``functor" to the natural principal $K_{\mathds{R}}$-bundle
\begin{equation}\label{eqliftablles08711}
\pi:G_{\mathds{R}}\to G_{\mathds{R}}/K_{\mathds{R}},
\end{equation}
we obtain the following $LK_{\mathds{R}}$-bundle:
\begin{equation}\label{eqliftabllesL08711}
L\pi:LG_{\mathds{R}}\to L(G_{\mathds{R}}/K_{\mathds{R}}).
\end{equation}
We also have the following tame Fr\'{e}chet principal $\mathds{C}^{*}K_{\mathds{R}}$-bundle
\begin{equation}\label{eqliftablebundles02}
\mathds{C}^{*}G_{\mathds{R}}\to\mathds{C}^{*}G_{\mathds{R}}/\mathds{C}^{*}K_{\mathds{R}}.
\end{equation}
Combining these bundles with the embeddings in Lemma~\ref{liftlopembed} and Lemma~\ref{freeloopfunctorgothru}, we obtain the following commutative diagram:
\begin{equation} \label{eqoveallfigsofar}
\xymatrix{
	\mathds{C}^{*}K_{\mathds{R}}\ar@{^{(}->}[d] \ar@{^{(}->}[r]& LK_{\mathds{R}} \ar@{^{(}->}[d]&\\
	\mathds{C}^{*}G_{\mathds{R}}\ar@{->>}[d] \ar@{^{(}->}[r]& LG_{\mathds{R}}\ar@{->>}[d]^{\Pi}\ar@{->>}[dr]^{L\pi}&\\
	\mathds{C}^{*}G_{\mathds{R}}/\mathds{C}^{*}K_{\mathds{R}} \ar@{^{(}->}[r]^{i}& LG_{\mathds{R}}/LK_{\mathds{R}} \ar@{^{(}->}[r]^{\tilde{\Phi}}&L(G_{\mathds{R}}/K_{\mathds{R}})}
\end{equation}
Now by lifting the bundle $L\pi$ over $\mathds{C}^{*}G_{\mathds{R}}/\mathds{C}^{*}K_{\mathds{R}}$ via $\tilde{\Phi}\circ i$ in (\ref{eqoveallfigsofar}) we obtain the following principal $LK_{\mathds{R}}$-bundle:
\begin{equation}\label{eqlltloeryp0}
\xi:~LK_{\mathds{R}}\hookrightarrow (\tilde{\Phi}\circ i)^{*}(L\pi)\to\mathds{C}^{*}G_{\mathds{R}}/\mathds{C}^{*}K_{\mathds{R}}.
\end{equation}
\end{proof}
\begin{remark}\label{remcontsetting}
	Note that any structure on $\Pi$ and, particularly, $L\pi$ can be lifted to $\xi$.\ Moreover, if we replace smooth loops by continuous loops and denote the mapping space it generates by $\mathfrak{L}$, we obtain a diagram similar to (\ref{eqoveallfigsofar}) but instead, the bundle (\ref{eqliftablebundles022}) becomes a locally trivial Banach principal $\mathfrak{L}K_{\mathds{R}}$-bundle\ (see \cite{zbMATH04002404}) as a special case of Proposition~\ref{propprincbundle}, namely
	\begin{equation}\label{eq:remlcouns02}
	\mathfrak{L}G_{\mathds{R}}\to \mathfrak{L}G_{\mathds{R}}/\mathfrak{L}K_{\mathds{R}}.
	\end{equation}	
\end{remark}
Let $\mathfrak{D}$ denote a Dynkin diagram and let $G(\mathfrak{D})$ denote the split (algebraically) simply-connected real Kac-Moody group of type $\mathfrak{D}$ (see \cite{zbMATH04017219}).\ By
$K(\mathfrak{D})$ we understand the fixed point subgroup associated to the
Cartan-Chevalley involution $\sigma_{\mathfrak{D}}$ on $G(\mathfrak{D})$ with respect to $\mathfrak{D}$.\ We also call $K(\mathfrak{D})$ a maximal compact form of $G(\mathfrak{D})$.\ Moreover, if $\mathfrak{D}$ is an irreducible simply laced Dynkin diagram, then by \cite[Theorem A]{zbMATH06710741}, up to isomorphism, there exists a uniquely determined group $Spin(\mathfrak{D})$ with a canonical two-to-one central extension:
\begin{equation}\label{eqcentralK}
\mathds{Z}_{2}\to\text{Spin}(\mathfrak{D})\stackrel{\rho}{\longrightarrow}K(\mathfrak{D}).
\end{equation}
For a spherical Dynkin diagram $\mathfrak{D}$, let $L^{\circ}K(\mathfrak{D})$ denote the connected component of $LK(\mathfrak{D})$.\ When
$K(\mathfrak{D})$ is connected, it is known that the number of connected components of
$LK(\mathfrak{D})$ is equal to the cardinality of the first homotopy group
of $K(\mathfrak{D})$.\ Recall that in this setting when $\mathfrak{D}$ is spherical, we always assume $K(\mathfrak{D})$
to be \textit{connected}.\\
Before we state the next result, notice that for a spherical irreducible
Dynkin diagram $\mathfrak{D}$ of rank $n\geq3$, if $BLK(\mathfrak{D})$ denotes the classifying space of $LK(\mathfrak{D})$ then by the universal coefficient theorem and the Hurewicz theorem we have
\begin{equation}\label{eqHurewiczthm}
\text{H}^{1}(BLK(\mathfrak{D});A)\cong\text{Hom}(\pi_{1}(K(\mathfrak{D})),A),
\end{equation}
for any abelian group $A$.\\
Here we give a slightly generalized version of \cite[Proposition 2.1]{zbMATH04216267} .
\begin{prop}\label{propliftconcom}
	Let $M$ be a 1-connected orientable Riemannian manifold of finite dimension.\
	Let $P\to M$ be a principal $K(\mathfrak{D})$-bundle where $\mathfrak{D}$ is a spherical irreducible Dynkin diagram of rank $n\geq3$ and let $\tilde{K}(\mathfrak{D})$ denote the universal covering group of $K(\mathfrak{D})$.\ Moreover, assume that $\pi_{1}(K(\mathfrak{D}))\cong\mathds{Z}_{2}$.\ Then the following are equivalent
	\begin{description}
		\item[(1)] The structure group of $LP\to LM$ is reducible to $L^{\circ}K(\mathfrak{D})$.
		\item[(2)] The structure group of $P\to M$ can be lifted to $\tilde{K}(\mathfrak{D})$.
		\item[(3)] The structure group of $LP\to LM$ can be lifted to $L\tilde{K}(\mathfrak{D})$.
	\end{description}
\end{prop}
\begin{proof}
	It follows from the proof of \cite[Proposition 2.1]{zbMATH04216267} by replacing $\text{SO}(n)$ by $K(\mathfrak{D})$, replacing $A$ by $\pi_{1}(K(\mathfrak{D}))$ in (\ref{eqHurewiczthm}) and using the following covering map:
	\begin{equation}\label{eqcentralpropliftconco}
	\pi_{1}(K(\mathfrak{D}))\to \tilde{K}(\mathfrak{D})\to K(\mathfrak{D}).
	\end{equation}
\end{proof}
\begin{cor}\label{corliftstruc}
	Let $\mathfrak{D}$ be a spherical irreducible Dynkin diagram of rank $n\geq3$.\ Assume that $G(\mathfrak{D})$ is compact and 1-connected.\ Assume that $K(\mathfrak{D})$ is connected and $\pi_{1}(K(\mathfrak{D}))\cong\mathds{Z}_{2}$.\ Then the primitive affine Kac-Moody symmetric space
	\begin{equation}\label{eqcorliftstruc00}
	\mathds{C}^{*}G(\mathfrak{D})/\mathds{C}^{*}K(\mathfrak{D}),
	\end{equation}
	admits an $L\tilde{K}(\mathfrak{D})$ structure if the structure group of the bundle
	\begin{equation}\label{eqcorliftstruc01}
	G(\mathfrak{D})\to G(\mathfrak{D})/K(\mathfrak{D}),
	\end{equation}
	can be lifted to $\tilde{K}(\mathfrak{D})$, or equivalently, if the structure group of
	\begin{equation}\label{eqcorliftstruc02}
	LG(\mathfrak{D})\to L(G(\mathfrak{D})/K(\mathfrak{D})),
	\end{equation}
	can be reduced to $L^{\circ}K(\mathfrak{D})$.
\end{cor}
\begin{proof}
	Since $G(\mathfrak{D})$ is 1-connected and $K(\mathfrak{D})$ is connected, by the long exact sequence of homotopy groups induced by (\ref{eqcorliftstruc01}) we conclude that $G(\mathfrak{D})/K(\mathfrak{D})$ is 1-connected.\ The corollary follows from Proposition~\ref{propliftconcom} and pulling back the $L\tilde{K}(\mathfrak{D})$ structure over (\ref{eqcorliftstruc02})
	on (\ref{eqcorliftstruc00}) via the commutative diagram (\ref{eqoveallfigsofar}) in Proposition~\ref{proLKbun}.
\end{proof}
\begin{definition}\label{deflspinstructure}
	For a spherical simply laced irreducible Dynkin diagram $\mathfrak{D}$, a lift of
	the structure group of a principal $LK(\mathfrak{D})$-bundle to $L\tilde{K}(\mathfrak{D})=L\text{Spin}(\mathfrak{D})$ (if exists) 
	is called a \textit{loop-spin structure}.
\end{definition}
The next results states that every loop-spin structure on a primitive affine Kac-Moody symmetric space can be pulled back of the corresponding affine Kac-Moody space in a natural manner.

\begin{lemma}\label{lemextendtoafkmsym}
	In the above setting, if the primitive affine Kac-Moody symmetric  space $\mathds{C}^{*}G(\mathfrak{D})/\mathds{C}^{*}K(\mathfrak{D})$
	admits a loop-spin structure, then the affine symmetric Kac-Moody space 
	$\widehat{\mathds{C}^{*}G(\mathfrak{D})}/\widehat{\mathds{C}^{*}K(\mathfrak{D})}$ 
	admits a loop-spin structure as well.
\end{lemma}

\begin{proof}
	Similar to the proof of \cite[Theorem 3.4.1]{WFrei09} adapted for affine Kac-Moody
	symmetric spaces of type I, $\widehat{\mathds{C}^{*}G(\mathfrak{D})}/\widehat{\mathds{C}^{*}K(\mathfrak{D})}$
	can be considered as an ($(\mathds{R}^{+})^{2}$ )-bundle over $\mathds{C}^{*}G(\mathfrak{D})/\mathds{C}^{*}K(\mathfrak{D})$
	(which, unlike the case in \cite[Theorem 3.4.1]{WFrei09}, does not need to be trivial).\ Hence
	there exists a surjection 
	\begin{equation}\label{eq:lemextendtoafkmsym00}
	P:\widehat{\mathds{C}^{*}G(\mathfrak{D})}/\widehat{\mathds{C}^{*}K(\mathfrak{D})}\to \mathds{C}^{*}G(\mathfrak{D})/\mathds{C}^{*}K(\mathfrak{D}). 
	\end{equation}
	The pull-back of the loop-spin structure on $\mathds{C}^{*}G(\mathfrak{D})/\mathds{C}^{*}K(\mathfrak{D})$ on $\widehat{\mathds{C}^{*}G(\mathfrak{D})}/\widehat{\mathds{C}^{*}K(\mathfrak{D})}$ by $P$
	induces a loop spin structure for $\widehat{\mathds{C}^{*}G(\mathfrak{D})}/\widehat{\mathds{C}^{*}K(\mathfrak{D})}$.
\end{proof}

\begin{definition}\label{deflstringstructure}
	For a spherical simply laced irreducible Dynkin diagram $\mathfrak{D}$, a lift of
	the loop-spin structure on a principal bundle to the universal extension $\widehat{L\text{Spin}(\mathfrak{D})}:=L\text{Spin}(\mathfrak{D})\rtimes\mathcal{S}^{1}$ 
	is called a \textit{loop-string structure}.\
\end{definition}

The next result provides a condition under which an affine Kac-Moody symmetric space admits a loop-string structure.
\begin{prop}\label{proplstringafkmsym}
	Let $\mathfrak{D}$ be a spherical simply laced irreducible Dynkin diagram such that $G(\mathfrak{D})$ is compact and 1-connected, and
	$K(\mathfrak{D})\cong\text{SO}(n)$ where $n\geq5$.\ Assume that the canonical principal $K(\mathfrak{D})$-bundle
	\[
	\xi:~~K(\mathfrak{D})\hookrightarrow G(\mathfrak{D})\to G(\mathfrak{D})/K(\mathfrak{D}),
	\]
	has a spin structure.\ If the first Pontrjagin class of $\xi$ is zero (i.e., $p_{1}(\xi)=0$), then
	$\widehat{\mathds{C}^{*}G(\mathfrak{D})}/\widehat{\mathds{C}^{*}K(\mathfrak{D})}$ admits a loop-string
	structure.	
\end{prop}
\begin{proof}
	Let 
	\[
	\eta:\text{Spin}(n)\hookrightarrow Q\to G(\mathfrak{D})/K(\mathfrak{D}),
	\]
	be the spin bundle associated to $\xi$.\ Since $p_{1}(\xi)=0$ by \cite[Theorem 3.1]{zbMATH04216267}
	\[
	L\eta:LQ\to L(G(\mathfrak{D})/K(\mathfrak{D})),
	\]
	admits a loop-string structure.\ Now by similar procedures as in the proof of Corollary~\ref{corliftstruc}
	and Lemma~\ref{lemextendtoafkmsym}, we can pull back the loop-string structure of $L(G(\mathfrak{D})/K(\mathfrak{D}))$ first on $\mathds{C}^{*}G(\mathfrak{D})/\mathds{C}^{*}K(\mathfrak{D})$
	and then on the corresponding affine Kac-Moody symetric space, namely,  $\widehat{\mathds{C}^{*}G(\mathfrak{D})}/\widehat{\mathds{C}^{*}K(\mathfrak{D})}$.
\end{proof}

\begin{example}[A class of loop-spin structures of type $E_{9}I$]\label{affineexample00}
	For $\mathfrak{D}=E_{8}$, we have $G(E_{8})$ is compact and 1-connected, also, for
	$K(E_{8})\cong (\text{SU}(2)\times E_{7})/\mathbb{Z}_{2}\cong \text{SO}(16)$. Hence by Corollary~\ref{corliftstruc}, any spin structure on $G(E_{8})/K(E_{8})$ can be pulled back on an affine Kac-Moody symmetric space of type $E_{9}I$, namely
	\begin{equation}\label{eqorliftstruc00}
	LSpin(\xi):  \widehat{L\text{Spin}(\mathfrak{D})}\hookrightarrow L\mathcal{S}\to \widehat{\mathds{C}^{*}G(E_{8})}/\widehat{\mathds{C}^{*}K(E_{8})}.
	\end{equation}
	According to the well-known list of finite dimensional symmetric spaces (see Table V of \cite[\S X.6.2]{zbMATH03705455}), a string structure of the symmetric space of type (E VIII) also provides a loop-string 
	structure on the above affine Kac-Moody symmetric space.\ The same stands for the affine Kac-Moody symmetric space obtained from the symmetric space of type (A I) for ranks greater than 4.
\end{example}

\section{Spin-c Representations of Type \texorpdfstring{$E_{n}$}{TEXT}}\label{secAKMSSCT}
A finite dimensional representation of $K(E_{10})$ in $\text{SO}(32)$ is presented by means of the Berman's generators and relations (see \cite{zbMATH04134274}) and investigated in \cite{zbMATH06088567,DaKlNi006,BuyHePa06,zbMATH02212653}.\ By the natural embedding of $G(E_{n})$ into $G(E_{10})$ for $6\leq n\leq9$, one can obtain a representation of $K(E_{n})$ in $\text{SO}(32)$ for $6\leq n\leq9$.\ This is specially interesting for us in the affine case $K(E_{9}).$\ Using this finite dimensional representation, we produce some methods to construct certain finite spin as well as $\text{spin}^{c}$ structures on affine Kac-Moody symmetric spaces of type $E_{n}.$\ This leads to providing sufficient conditions for lifting the $\text{SO}(32)$ representation of certain $K$ to a $\text{Spin}^{c}(32)$ representation. 

Let $M:=G/K$ be an affine Kac-Moody symmetric space of compact type I constructed in \cite[\S 4.3.3]{WFrei09} where $G:=\widehat{\mathds{C}^{*}G^{\circ}_{\mathds{R}}}$, $G^{\circ}_{\mathds{R}}$ are of type $E_{n}$ (for $6\leq n\leq 8$) and $K:=\widehat{\mathds{C}^{*}K^{\circ}_{\mathds{R}}}$ are as in Section~\ref{secELB} obtained from a Cartan-Chevalley involution.\  We saw in
Section~\ref{secFacSpPB} that 
\begin{equation}\label{eq:pi000}
\pi:G\to M,
\end{equation}
is a locally trivial tame Fr\'echet principal $K$-bundle.\ This local triviality obtained in Theorem~\ref{maintheorel}, enables us to work with locally trivial associated bundles.\ Using the {\v C}ech cohomology corresponding to locally trivial principal bundle, in this section we present conditions under which the above bundle admits finite dimensional spin bundles by the finite dimensional representation of $K(E_{10})$ .\ 

First we start by showing that when $K^{\circ}_{\mathds{R}}$ is semisimple, then $K:=\widehat{\mathds{C}^{*}K^{\circ}_{\mathds{R}}}$ has a finite dimensional representation induced from the $\text{SO}(32)$ representation of $K(E_{10}).$\ Recall from the list of finite dimensional symmetric spaces (see Table V of \cite[\S X.6.2]{zbMATH03705455}) that in types $E_{n}$ (for $6\leq n\leq 8$), except for cases (E III) and (E VIII),  $K(E_{n})$ is semisimple.\ Therefore, the semisimple assumption is not very restrictive.

Recall from \cite[Section 2.1]{Pope005} that $\widehat{LK^{\circ}_{\mathbb{R}}}\cong\mathcal{S}^{1}\ltimes\widetilde{LK^{\circ}_{\mathbb{R}}}$  where $\widetilde{LK^{\circ}_{\mathbb{R}}}$ is the universal central extension of $LK^{\circ}_{\mathbb{R}}$.

\begin{prop}\label{proprepofK}
	Let $K(E_{9})$ be the fixed point subgroup of a Cartan-Chevalley involution compatible with the Cartan-Chevalley involution corresponding to $K^{\circ}_{\mathbb{R}}$ for the loop representation	of $G(E_{9})$.\ Let $K^{\circ}_{\mathds{R}}$ be semisimple. Then  $K:=\widehat{\mathds{C}^{*}K^{\circ}_{\mathds{R}}}$ has a finite dimensional
	(unfaithful) $\text{SO}(32)$-representation.\ This representation splits into an $\text{SO}^{+}(16)\times\text{SO}^{-}(16)$-representation.
\end{prop}

\begin{proof}
	First we show that the $\text{SO}(32)$ representation of $K(E_{9})$ extends continuously to $\widehat{LK^{\circ}_{\mathbb{R}}}$:
	
	$K(E_{9})$ (embedded in $K(E_{10})$) admits a finite dimensional (unfaithful) representation on $\text{SO}(32)$.\ This representation, by the main result of \cite{KNP07},
	splits into $\text{SO}^{+}(16)\times\text{SO}^{-}(16)$. Namely
	\begin{equation}\label{eq:lemrepofK00}
	\rho:K(E_{9})\to \text{SO}^{+}(16)\times\text{SO}^{-}(16)\hookrightarrow\text{SO}(32).
	\end{equation}
	The $\pm$ signs here are superficial.\ Note that the loop representation of $K(E_{9})$ is of the form $\widehat{L_{pol }K^{\circ}_{\mathbb{R}}}\cong\mathcal{S}^{1}\ltimes\widetilde{LK^{\circ}_{\mathbb{R}}}$ (see e.g., \cite[Chap.'s 7 \& 8]{zbMATH00194085}) where $\widetilde{L_{pol }K^{\circ}_{\mathbb{R}}}$ is the universal central extension of $L_{pol }K^{\circ}_{\mathbb{R}}$ (see \cite[Section 13.2]{zbMATH01821146} or \cite{Pope005}).\
	Since $K^{\circ}_{\mathds{R}}$ is semisimple, $L_{pol}K^{\circ}_{\mathbb{R}}$ is dense in $LK^{\circ}_{\mathbb{R}}$ by \cite[Proposition 3.5.3]{zbMATH04002404}.\
	Hence $K(E_{9})$ is dense in $\widehat{LK^{\circ}_{\mathbb{R}}}\cong\mathcal{S}^{1}\ltimes\widetilde{LK^{\circ}_{\mathbb{R}}}$.\ Therefore,
	the representation $\rho$ continuously extends to an (unfaithful) representation (denoted again by $\rho$) of $\widehat{LK^{\circ}_{\mathbb{R}}}$ (which preserves the $\text{SO}^{+}(16)\times\text{SO}^{-}(16)$-split since $\text{SO}^{\pm}(16)$ are closed subgroups), namely,
	\begin{equation}\label{eq:lemrepofK01}
	\rho:\widehat{LK^{\circ}_{\mathbb{R}}}\to\text{SO}^{+}(16)\times\text{SO}^{-}(16)\hookrightarrow \text{SO}(32).
	\end{equation}
	
	Next we show that $K$ can be continuously embedded in $\widehat{LK^{\circ}_{\mathbb{R}}}$: 
	
	Recall from Lemma~\ref{lemcontembpre} that there exists the following continuous injective morphism of topological groups:  
	\[
	H:\mathds{C}^{*}K^{\circ}_{\mathds{R}} \hookrightarrow LK^{\circ}_{\mathds{R}}. 
	\]
	Moreover, $K$ can be considered as a torus bundle over $\mathds{C}^{*}K^{\circ}_{\mathds{R}}$ compatible with $\widehat{LK^{\circ}_{\mathbb{R}}}$ when considered as a torus bundle over ${LK^{\circ}_{\mathbb{R}}}$ (see \cite[\S 4.3.2]{WFrei09}).\ Therefore, $H$ extends continuously and invectively to $K$ denoted again by $H$, namely,
	\begin{equation}\label{eqproprepofK}
	H:K \hookrightarrow \widehat{LK^{\circ}_{\mathbb{R}}}.
	\end{equation} 
	Hence, the proposition follows from the combination
	
	\begin{equation}\label{eqproprepofKK}
	\varrho:=H\circ\rho:K\to\text{SO}^{+}(16)\times\text{SO}^{-}(16)\hookrightarrow \text{SO}(32).
	\end{equation} 	
\end{proof} 

From now on, we always assume that $K^{\circ}_{\mathds{R}}$ is semisimple.\ The associated bundle to $\pi$ with respect to $\varrho$ is defined as follows
\begin{equation}\label{fframebundle}
\text{SO}(\pi):G\times_{\varrho}\text{SO}(32)\to M.
\end{equation}
$\text{SO}(\pi)$ is called the \textit{finite frame bundle} on $M$ with respect to $\varrho$.\ 
\begin{definition}\label{deffspinstr}
	We say $M$ has a \textit{finite $\text{spin}$ ($\text{spin}^{c}$) structure} if the structure group of $\text{SO}(\pi)$ has a lift to a  $\text{Spin}(32)$ ($\text{Spin}^{c}(32)$).\ When $M$ has a finite spin structure, define
	\[
	\text{Spin}(\pi):=G\times_{\tilde{\varrho}}\text{Spin}(32)\to M.
	\]
	We call the following associated bundle, the \textit{finite spinor bundle} over $M$:
	\begin{equation}\label{eqdeffspinstr}
	\mathcal{S}(M):=G\times_{\tilde{\varrho}}\mathbb{S}^{32}\to M,
	\end{equation}
	where $\text{Clif}(\mathbb{R}^{32})\otimes\mathbb{C}\cong\text{End}(\mathbb{S}^{32})$.\ Similarly,
	when $M$ admits a finite $\text{Spin}^{c}$ structure, define
	\[
	\text{Spin}^{c}(\pi):=G\times_{\tilde{\varrho}}\text{Spin}^{c}(32)\to M.
	\]
	We call the associated bundle, the \textit{finite spinor-c bundle} over $M$ denoted by $\mathcal{S}^{c}(M)$.
	\end{definition}
\begin{remark}\label{remexspincstr}\leavevmode
	\begin{itemize}
		\item[(I)] In the context of {\v C}ech cohomology theory, since $\text{Ad}:\text{Spin}(32)\to\text{SO}(32)$ is a universal covering map, having a finite spin structure is equivalent to vanishing the induced following morphism when $K$ is connected
		\[
		\varrho_{\ast}:\pi_{1}(K)\to\pi_{1}(SO(32))\cong\mathbb{Z}_{2},
		\]
		(vanishing $\rho_{\ast}$ implies vanishing the second Stiefel-Whitney class of $SO(\pi)$).\ Note that by the definition of $\varrho$ in Proposition~\ref{proprepofK}, to have a $\text{spin}$ structure it is enough to verify the above equivalent	condition for $\widehat{LK^{\circ}_{\mathbb{R}}}$.
		\item[(II)] It is also important to find a condition under which
		the structure group $K$ of $\pi:G\to M$ reduces to its one-component $K^{\circ}$.\ Note that since $K$ is a topological
		group, we have the following exact sequence
		\begin{equation}\label{eq:remlifttoonecom00}
		K^{\circ}\to K\to\pi_{0}(K).
		\end{equation}
		Accordingly, the following {\v C}ech cohomological mapping exits:
		\begin{equation}\label{eq:remlifttoonecom01}
		\nu_{1}:\text{H}^{1}(M,K)\to\text{H}^{1}(M,\pi_{0}(K)).
		\end{equation}  
		Hence, the structure group $K$ of $\pi:G\to M$ reduces to $K^{\circ}$
		if and only if $\nu_{1}(\pi)=0$.
	\end{itemize}
	
\end{remark}

\begin{theorem}\label{thmspincadm}
	If $K^{\circ}_{\mathbb{R}}$ simply connected then $M$ admits a finite $\text{spin}^{c}$ structure.
\end{theorem}
\begin{proof}
	Recall from (\ref{eqproprepofK}) in Proposition~\ref{proprepofK} that there exists the following continuous injective morphism of topological group: 
	\[
	H:K \hookrightarrow LK^{\circ}_{\mathds{R}}. 
	\]
	Let
	\begin{equation}\label{eq:remlifttoonecom05}
	\pi^{H}:G\times_{H}\widehat{LK^{\circ}_{\mathbb{R}}}\to M,
	\end{equation}
	be the associated (topological) bundle to $\pi$ with respect to $H$.\ Now, with the same argument as in Remark~\ref{remexspincstr}, having a $\text{spin}^{c}$ structure is equivalent to the condition that $\pi_{1}(K)$ does not contain 2-torsion elements.\ Therefore, in view of (\ref{eq:remlifttoonecom05}) and the definition of $\varrho$ in Proposition~\ref{proprepofK}, to have a $\text{spin}^{c}$ structure it suffices to check the above equivalent torsion condition for $\widehat{LK^{\circ}_{\mathbb{R}}}$.
	
	Now note that since $\widehat{LK^{\circ}_{\mathbb{R}}}\cong\mathcal{S}^{1}\ltimes\widetilde{LK^{\circ}_{\mathbb{R}}}$, we have  
	\begin{equation}\label{eq:lemhomto01}
	\pi_{0}(\widehat{LK^{\circ}_{\mathbb{R}}})=\pi_{0}(\widetilde{LK^{\circ}_{\mathbb{R}}}),
	\end{equation}
	and 
	\begin{equation}\label{eq:lemhomto011}
	\pi_{1}(\widehat{LK^{\circ}_{\mathbb{R}}})=\mathbb{Z}\oplus\pi_{1}(\widetilde{LK^{\circ}_{\mathbb{R}}}).
	\end{equation}
	But by \cite[Proposition 4.6.9]{zbMATH04002404} we have 
	\begin{equation}\label{eq:lemhomto02}
	\pi_{0}(\widetilde{LK^{\circ}_{\mathbb{R}}})\cong\pi_{0}(\widetilde{LK^{\circ}_{\mathbb{R}}})\cong\pi_{1}(K^{\circ}_{\mathbb{R}}).
	\end{equation}
	Therefore, by the simply connected assumption, all $\widehat{LK^{\circ}_{\mathbb{R}}}, \widetilde{LK^{\circ}_{\mathbb{R}}} $ and $LK^{\circ}_{\mathbb{R}}$ are connected.\

	Now we show that $\pi_{1}(\widehat{LK^{\circ}_{\mathbb{R}}})$ is torsion free.\ Considering the $\mathcal{S}^{1}$-fibration 
	\begin{equation}\label{eq:protorsion00}
	\mathcal{S}^{1}\to\widetilde{LK^{\circ}_{\mathbb{R}}}\to LK^{\circ}_{\mathbb{R}},
	\end{equation}
	we obtain the following exact sequence of homotopy groups
	\begin{equation}\label{eq:protorsion01}
	0\to\pi_{2}(\widetilde{LK^{\circ}_{\mathbb{R}}})\to \pi_{2}(LK^{\circ}_{\mathbb{R}})\stackrel{\alpha}{\to}\pi_{1}(\mathcal{S}^{1})\to\pi_{1}(\widetilde{LK^{\circ}_{\mathbb{R}}})\to\pi_{1}(LK^{\circ}_{\mathbb{R}})\to 0.
	\end{equation}
	Recall that $\widetilde{LK^{\circ}_{\mathbb{R}}}$ is the universal central extension
	of $LK^{\circ}_{\mathbb{R}}$ (see \cite[Section 4.5.1]{WFreyn11} and \cite[Section 4.4]{zbMATH04002404}), therefore, its associated Chern class corresponds to a generator and hence the ``transgression'' map $\alpha$
	is surjective.\ Moreover, since $K^{\circ}_{\mathbb{R}}$ is simply connected, the ``transgression'' map $\alpha$
	is actually and isomorphism (see e.g., \cite{zbMATH04002404} and the proof of \cite[Proposition 2.1]{zbMATH04216267}).\ Hence form (\ref{eq:protorsion01})  we deduce the following
	exact sequence:
	\begin{equation}\label{eq:protorsion02}
	0\to\pi_{1}(\widetilde{LK^{\circ}_{\mathbb{R}}})\to\pi_{1}(LK^{\circ}_{\mathbb{R}})\to 0.
	\end{equation}
	
	This, in view of (\ref{eq:lemhomto011}), implies that $\pi_{1}(\widehat{LK^{\circ}_{\mathbb{R}}})$ is torsion free
	if and only if $\pi_{1}(\widetilde{LK^{\circ}_{\mathbb{R}}})$ is torsion free
	if and only if $\pi_{1}(K^{\circ}_{\mathbb{R}})$ is torsion free.\ The theorem follows.
\end{proof}

By relaxing the simply connected condition on $K^{\circ}_{\mathbb{R}}$ one can obtain certain conditions under which $M$ admits a finite $\text{spin}^{c}$ structure.\ These condition are more technical and more difficult to be verified.\ First we need to be able to reduce the structure group ${\widehat{LK^{\circ}_{\mathbb{R}}}}^{\circ}$ to its identity component.\ For this, similar to Remark~\ref{remexspincstr}(II), we
consider the following exact sequence 
\begin{equation}\label{eq:remlifttoonecom022}
{\widehat{LK^{\circ}_{\mathbb{R}}}}^{\circ}\to \widehat{LK^{\circ}_{\mathbb{R}}}\to\pi_{0}(\widehat{LK^{\circ}_{\mathbb{R}}}).
\end{equation}
This, in view of (\ref{eq:lemhomto02}), yields the following exact sequence 
\begin{equation}\label{eq:remlifttoonecom03}
{\widehat{LK^{\circ}_{\mathbb{R}}}}^{\circ}\to \widehat{LK^{\circ}_{\mathbb{R}}}\to\pi_{1}(K^{\circ}_{\mathbb{R}}),
\end{equation}
which induces the following cohomological mapping
\begin{equation}\label{eq:remlifttoonecom04}
\tilde{\nu}_{1}:\text{H}^{1}(M,\widehat{LK^{\circ}_{\mathbb{R}}})\to\text{H}^{1}(M,\pi_{0}(\widehat{LK^{\circ}_{\mathbb{R}}})).
\end{equation} 
Hence the structure group $\widehat{LK^{\circ}_{\mathbb{R}}}$ reduces to ${\widehat{LK^{\circ}_{\mathbb{R}}}}^{\circ}$ if and only if $\tilde{\nu}_{1}(\pi^{H})=0$.\ Now by similar argument as in the proof of the above theorem one can deduce the following:
\begin{cor}\label{cortorsion}
	In  the above setting, if 
	\begin{itemize}
		\item $\tilde{\nu}_{1}(\pi)=0$; and,
		\item the first Chern class of $\widetilde{LK^{\circ}_{\mathbb{R}}}$ corresponds to
		a generator; and,
		\item $\pi_{1}(K^{\circ}_{\mathbb{R}})$ is torsion free,
	\end{itemize}
	then $M$ admits a finite $\text{Spin}^{c}$ structure. 
\end{cor}

\begin{example}\label{exaCartanEIV}
	Let $M:=\widehat{\mathds{C}^{*}(E_{6})_{\mathds{R}}}/\widehat{\mathds{C}^{*}(F_{4})_{\mathds{R}}}$ be the affine symmetric space obtained from the symmetric space of type E IV on the Cartan's list of symmetric spaces.\ Since $(F_{4})_{\mathds{R}}$ is simply connected, $M$ admits a finite $\text{spin}^{c}$ structure.\ Another interesting example is type (E IX) where $M:=\widehat{\mathds{C}^{*}(E_{8})_{\mathds{R}}}/\widehat{\mathds{C}^{*}(E_{7}\times\text{SU}(2))_{\mathds{R}}}$.
\end{example}

\begin{cor}
	Let $M$ be an affine Kac-Moody symmetric space with a finite $\text{Spin}^{c}$ structure.\ Then the $\text{SO}(32)$ representation $\rho$ of the corresponding algebraic Kac-Moody ``compact'' forms can be lifted to $\text{Spin}^{c}(32)$.\ Moreover, such affine Kac-Moody symmetric spaces exist.
\end{cor}

\begin{proof}
	First note that $\tilde{\varrho}$ is a lift of the $\text{SO}(32)$ representation $\varrho$ of $K$ (see (\ref{eqproprepofKK})) to $\text{Spin}^{c}(32)$.\ Now since $K=\widehat{\mathds{C}^{*}K^{\circ}_{\mathds{R}}}$ contains the polynomial loop group $\widehat{L_{pol }K^{\circ}_{\mathbb{R}}}$, having a finite $\text{Spin}^{c}$ structure means obtaining a finite spin representation for certain algebraic Kac-Moody "compact" forms.\ The corollary follows from Theorem~\ref{thmspincadm}.  
\end{proof}

\bibliographystyle{plain}     
\bibliography{document}   

\end{document}